\documentclass[runningheads]{llncs}
\usepackage[T1]{fontenc}
\usepackage{graphicx}
\usepackage{amssymb}
\usepackage{amsmath}
\usepackage{todonotes}
\usepackage{comment}
\usepackage{algorithm}
\usepackage{algpseudocode}
\algnewcommand\And{\textbf{ and }}
\algnewcommand\Or{\textbf{ or }}
\algnewcommand\Elif{\newline\textbf{else if }}
\algnewcommand\Then{\textbf{ then}}

\usepackage{tikz}
\usetikzlibrary{quantikz}

\newcommand{\Endproof}{\hfill$\Box$\\}

\usepackage{subcaption}

\begin{document}
\title{Quantum Algorithm for the Multiple String Matching Problem}
\author{Kamil Khadiev\inst{1}\orcidID{0000-0002-5151-9908}\and Danil Serov\inst{1}}
\authorrunning{K. Khadiev and D. Serov}
\institute{
Kazan Federal University, Kazan, Russia  
\\
\email{kamilhadi@gmail.com}
}
\maketitle              % typeset the header of the contribution
\begin{abstract}
Let us consider the Multiple String Matching Problem. In this problem, we consider a long string, denoted by $t$, of length $n$. This string is referred to as a text. We also consider a sequence of $m$ strings, denoted by $S$, which we refer to as a dictionary. The total length of all strings from the dictionary is represented by the variable L. The objective is to identify all instances of strings from the dictionary within the text. The standard classical solution to this problem is Aho–Corasick Algorithm that has $O(n+L)$ query and time complexity. At the same time, the classical lower bound for the problem is the same $\Omega(n+L)$.  We propose a quantum algorithm  with $O(n+\sqrt{mL\log n}+m\log n)$ query complexity and $O(n+\sqrt{mL\log n}\log b+m\log n)=O^*(n+\sqrt{mL})$ time complexity, where $b$ is the maximal length of strings from the dictionary. This improvement is particularly significant in the case of dictionaries comprising long words. Our algorithm's complexity is equal to the quantum lower bound $O(n + \sqrt{mL})$, up to a log factor.  In some sense, our algorithm can be viewed as a quantum analogue of the Aho–Corasick algorithm. 

\end{abstract}
\section{Introduction}
\label{sec:intro}
Let us consider the Multiple String Matching Problem. In this problem, we have a long string $t$ of a length $n$ that we call a text; and a sequence of $m$ strings $S$ that we call a dictionary. The total length of all strings from the dictionary is $L$, and the maximal length of the dictionary's strings is $b$. Our goal is to find all positions of strings from the dictionary in the text. The standard classical solution to this problem is Aho–Corasick Algorithm \cite{ac75} that has $O(n+L)$ time and query complexity. 

We are interested in a quantum algorithm for this problem. There are many examples of quantum algorithms \cite{nc2010,a2017,k2022lecturenotes,aazksw2019part1} that are faster than classical counterparts  \cite{dw2001,quantumzoo}. 
Particularly, we can find such examples among problems for strings processing   \cite{aaksv2022,aj2021,kkmsy2022,ki2019,kbcw2024,ke2022,kiv2022,kk2021,kb2022,kr2021b,kr2021a,kszm2022,l2020,l2020conf,m2017,ks2024}.
%add polindrom 
One such problem is the String matching problem. That is checking whether a string $s$ is a substring of a string $t$. It is very similar to the Multiple String Matching Problem but with only one string in the dictionary. Additionally, in the String matching problem, we should find any position of $s$ not all of them. The last difference can be significant for quantum algorithms. For this problem, we can use, for example, Knuth–Morris–Pratt algorithm \cite{cormen2001,kmp77} in the classical case with $O(n+k)$ time and query complexity, where $n$ is the length of $t$ and $k$ is the length of $s$. Other possible algorithms can be Backward Oracle Matching \cite{acr1999}, Wu-Manber \cite{wm1994}.  In the quantum case, we can use Ramesh-Vinay algorithm \cite{rv2003} with $O\left(\sqrt{n}\log{\sqrt{\frac{n}{k}}}\log n + \sqrt{k}(\log k)^2\right)=$ query complexity. We can say that query and time complexity of the Ramesh-Vinay algorithm is $O^*(\sqrt{n}+\sqrt{k})$, where $O^*$ notation hides log factors. If we naively apply these algorithms to the Multiple String Matching Problem, then we obtain $O(mn+L)$ complexity in the classical case, and $O\left(mn\log{\sqrt{\frac{n}{b}}}\log n + m^{1.5}\sqrt{L}(\log b)^2\right)=O^*(mn+m^{1.5}\sqrt{L})$ in the quantum case.

In the classical case, Aho–Corasick Algorithm \cite{ac75} has better complexity that is $O(n+L)$. Another possible algorithm that can work better in average is  the Commentz-Walter algorithm \cite{cw1979}. In the quantum case, we suggest a new algorithm for the Multiple String Matching Problem with $O(n+\sqrt{mL\log n}+m\log n)$   query complexity and $O(n+\sqrt{mL\log n}\log b+m\log n)$ time complexity. So, we can say that both are $O^*(n+\sqrt{mL})$. At the same time, we show that the lower bound for classical complexity is $\Omega(n+l)$, and for quantum complexity, it is $\Omega(n+\sqrt{mL})$. It means, that the Aho–Corasick Algorithm reaches the lower bound, and our algorithm reaches the lower bound up to a log factor. At the same time, if $O(m)$ strings of the dictionary have at least $\omega(\log n)$ length, then our quantum algorithm shows speed-up.
 
The structure of this paper is as follows:
Section \ref{sec:prelims} contains preliminaries, used data structures and algorithms.
Section \ref{sec:aho–corasick} provides a quantum algorithm for the Multiple String Matching Problem. The final Section \ref{sec:concl} concludes the paper and contains open questions. 

\section{Preliminaries and Tools}
\label{sec:prelims}
For a string $u=(u_1,\dots,u_M)$, let $|u|=M$ be a length of the string. Let $u[i:j]=(u_i,\dots,u_j)$ be a substring of $u$ from $i$-th to $j$-th symbol. Let $u[i]=u_i$ be $i$-th symbol of $u$. In the paper, we say that a string $u$ is less than a string $v$ if $u$ precedes $v$ in the lexicographical order.

\subsection{The Multiple String Matching Problem}
Suppose, we have a string $t$ and a sequence of strings $S=(s^1,\dots,s^m)$. We call the set $S$ a dictionary, and the the string $t$ a text. Let $L=|s^1|+\dots+|s^m|$ be the total length of the dictionary, and $b=\max\{|s^1|,\dots,|s^m|\}$ be the maximal length of dictionary's strings. Let an integer $n=|t|$ be the length of $t$. The problem is to find positions in $t$ for all strings from  $S$. Formally, we want to find a sequence of indexes ${\cal I}=(I_1,\dots,I_m)$, where $I_j=(i_{j,1},\dots,i_{j,k_j})$ such that $t[i_{j,x}:i_{j,x}+|s^j|-1]=s^j$ for each $x\in\{1,\dots,k_j\}, j\in\{1,\dots,m\}$. Here $k_j$ is the number of occurrences of $s^j$ in the string $t$.
%\subsection{The Problem for LZ78}

%In the paper, we use a {\em trie} (prefix tree) data structure \cite{d59,b98,b2008,knuth73}. It is a tree that allows us to add a string $s$  and check whether $s$ is in the tree with running time $O(|s|)$. The data structure implements a \``set of strings'' data structure. Let us have the following operations with a trie $T$:
%\begin{itemize}
%    \item $\textsc{InitTrie()}$ returns an empty trie. The running time of the operation is $O(1)$.
%    \item $\textsc{AddToTrie(T,s)}$ adds a string $s$ to the trie $T$. The running time of the operation is $O(|s|)$.
%    \item $\textsc{Contains(T,s)}$ returns $1$ if a string $s$ belongs to the trie $T$, and $False$ otherwise. The running time of the operation is $O(|s|)$.
%\end{itemize}

\subsection{Suffix Array}
A suffix array \cite{mm90} is an array $suf=(suf_1,\dots,suf_{l})$ for a string $u$ where $l=|u|$ is the length of the string. The suffix array is the lexicographical order for all suffixes of $u$. Formally, $u[suf_i:l]<u[suf_{i+1}:l]$ for any $i\in\{1,\dots,l-1\}$.
Let $\textsc{ConstructSuffixArray}(u)$ be a procedure that constructs the suffix array for the string $u$. The query and time complexity of the procedure are as follows:
\begin{lemma}[\cite{llh2018}]\label{lm:suf-arr}
A suffix array for a string $u$ can be constructed with $O(|u|)$ query and time complexity.
\end{lemma}
Let longest common prefix (LCP) of two strings $u$ and $s$ or $LCP(u,s)$ be an index $i$ such that $u_j=s_j$ for each $j\in\{1,\dots,i\}$, and $u_{i+1}\neq s_{i+1}$.
 
For a suffix array $suf$ we can define an array of longest common prefixes (LCP array). Let call it $lcp=(lcp_{1},\dots, lcp_{l-1})$, where $lcp_i=LCP(u[suf_i:l],u[suf_{i+1}:l])$ is the longest common prefix of two suffixes $u[suf_i:l]$ and $u[suf_{i+1}:l]$. The array can be constructed with linear complexity:
%Let $\textsc{ConstructLCPforSuffixArray}(u)$ be a procedure that constructs the LCP array for the string $u$. The query and time complexity of the procedure are as follows:
\begin{lemma}[\cite{lklaap2001,klaap2001}]\label{lm:lcp-arr}
An array of longest common prefixes (LCP array) for a string $u$ can be constructed with $O(|u|)$ query and time complexity.
\end{lemma}

Using the lcp array, we can compute the LCP of any two suffixes $LCP(u[suf_i:l],u[suf_{j}:l])$ in constant time due to the algorithm from \cite{bfc2000}.
\begin{lemma}[\cite{bfc2000}]\label{lm:lcp2-arr}
For $i,j\in\{1,\dots |u|\}$, there is an algorithm that computes the longest common prefix for suffixes $u[suf_i:l]$ and $u[suf_{j}:l]$ of the string $u$ with $O(1)$ query and time complexity with preprocessing with $O(|u|)$ query and time complexity. 
\end{lemma}
For $i,j\in\{1,\dots |u|\}$, let $\textsc{LCPSuf}(u,i,j)$ be the procedure that returns $LCP(u[suf_i:l],u[suf_{j}:l])$, and $\textsc{PreprocessingForLCP}(u)$ be the preprocessing procedure for this algorithm and constructing the $lcp$ array.

\subsection{Quantum Query Model}
One of the most popular computation models for quantum algorithms is the query model.
We use the standard form of the quantum query model. %is a generalization of the decision tree model of classical computation that is commonly used to lower bound the amount of time required by a computation. 
Let $f:D\rightarrow \{0,1\},D\subseteq \{0,1\}^M$ be an $M$ variable function. Our goal is to compute it on an input $x\in D$. We are given oracle access to the input $x$, i.e. it is implemented by a specific unitary transformation usually defined as $\ket{i}\ket{z}\ket{w}\mapsto \ket{i}\ket{z+x_i\pmod{2}}\ket{w}$, where the $\ket{i}$ register indicates the index of the variable we are querying, $\ket{z}$ is the output register, and $\ket{w}$ is some auxiliary work-space. An algorithm in the query model consists of alternating applications of arbitrary unitaries which are independent of the input and the query unitary, and a measurement at the end. The smallest number of queries for an algorithm that outputs $f(x)$ with probability $\geq \frac{2}{3}$ on all $x$ is called the quantum query complexity of the function $f$ and is denoted by $Q(f)$.
We refer the readers to \cite{nc2010,a2017,aazksw2019part1,k2022lecturenotes} for more details on quantum computing. 

In this paper, we are interested in the query complexity of the quantum algorithms. We use modifications of Grover's search algorithm \cite{g96,bbht98} as quantum subroutines. For these subroutines, time complexity can be obtained from query complexity by multiplication to $O(\log D)$, where $D$ is the size of a search space \cite{ad2017,g2002}.  

%\subsection{Noisy Self-Balanced Binary Search Tree for Strings with Quantum Comparator}
%The Noisy Self-Balanced Binary Search Tree for Strings with Quantum Comparator is a data structure that was introduced and developed in \cite{kszm2023,ke2022}. It is a Self-Balanced Binary Search Tree (Red-Black tree \cite{cormen2001} as an example) that stores strings (not strings themself but their indexes), and compares them using a quantum algorithm for comparing strings \cite{bjdb2017,ki2019,kiv2022}
%todo update citation \cite{bjdb2017,ki2019,a2019,kiv2022,kkmsy2022,kl2020,l2020} 
% and uses the Walking Tree technique \cite{kszm2022} for removing additional complexity because of comparing the procedure's error probability. The quantum algorithm for comparing strings is based on a modification of Grover's search algorithm \cite{g96,bbht98} that finds the minimum argument satisfied a predicate \cite{dhhm2004,k2014,ll2016,kkmsy2022}.
%todo update citation \cite{dhhm2004,k2014,ll2015,ll2016,kkmsy2022}

\section{Quantum Algorithm for the Multiple String Matching Problem}
\label{sec:aho–corasick}
Let us present a quantum algorithm for the Multiple String Matching Problem. A typical solution for the problem is the Aho–Corasick Algorithm that works with $O(n+L)$ query and time complexity. Here we present an algorithm that works with $O(n+\sqrt{mL\log n}+m\log n)$ query complexity. The algorithm uses ideas from the paper of Manber and Myers \cite{mm90} and uses quantum strings comparing algorithm \cite{bjdb2017,ki2019,kiv2022}.
%todo update citation \cite{bjdb2017,ki2019,a2019,kiv2022,kkmsy2022,kl2020,l2020} 
%that is based on modification of Grover's search algorithm \cite{g96,bbht98} that finds the minimum argument satisfied a predicate \cite{dhhm2006,k2014,ll2016,kkmsy2022}.
%todo update citation \cite{dhhm2006,k2014,ll2015,ll2016,kkmsy2022}

Firstly, let us construct a suffix array, LCP array, and invoke preprocessing for $\textsc{LCPSuf}$ function for the string $t$ (see Section \ref{sec:prelims} for details). Note that $n=|t|$.

Let us consider a string $s^j$ for $j\in\{1,\dots, m\}$. If $s^j$ is a substring of $t$ starting from an index $i$, then $s^j$ is a prefix of the suffix $t[suf_i:n]$, i.e. $t[suf_i:suf_i+|s^j|-1]=s^j$. Since $t[suf_i:n]$  strings are in a lexicographical order according to $suf$ order, all suffixes that have $s^j$ as prefixes are situated sequentially. Our goal is to find two indexes $left_j$ and $right_j$ such that all suffixes $t[suf_i:n]$ has $s^j$ as a prefix for each $i\in\{suf_{left_j},\dots,suf_{right_j}\}$. So, we can say that $I_j=(suf_{left_j},\dots,suf_{right_j})$.

There is a quantum algorithm for computing LCP of $u$ and $v$. Let $\textsc{QLCP}(u,v)$ be the corresponding function. Its complexity is presented in the next lemma. 
\begin{lemma}[\cite{kkmsy2022}]\label{lm:qlcp}
There is a quantum algorithm that implements $\textsc{QLCP}(u,v)$  procedure, and has query complexity $O(\sqrt{d})$ and time complexity $O(\sqrt{d}\log d)$, where $d$ is the minimal index of an unequal symbol of $u$ and $v$. The error probability is at most $0.1$. 
\end{lemma}
The algorithm is an application of the First-One Search algorithm \cite{dhhm2006,k2014,ll2016,kkmsy2022} that finds a minimal argument of a Boolean-valued function (predicate) that has the $1$-result. The First-One Search algorithm algorithm has an error probability $0.1$ that is why we have the same error probability for $\textsc{QLCP}$ procedure. The First-One Search is an algorithm that is based on the Grover Search algorithm \cite{g96,bbht98}. So, the difference between query and time complexity is a log factor due to \cite{ad2017,g2002}. The $\textsc{QLCP}(u,v)$ procedure is similar to the quantum algorithm that compares two strings $u$ and $v$ in lexicographical order that was independently developed in \cite{bjdb2017} and in \cite{ki2019,kiv2022}.

For $i\in\{1,\dots,min(|u|,|v|)\}$, let $\textsc{QLCP}(u,v,i)$ be a similar function that computes LCP starting from the $i$-th symbol. Formally, $\textsc{QLCP}(u,v,i)=\textsc{QLCP}(u[i:|u|],v[i:|v|])$. Due to Lemma \ref{lm:qlcp}, complexity of $\textsc{QLCP}(u,v,i)$ is following
\begin{corollary}\label{cr:lcp}
There is a quantum algorithm that implements $\textsc{QLCP}(u,v,i)$ procedure, and has $O(\sqrt{d-i})$ query complexity and $O(\sqrt{d-i}\log (d-i))$ time complexity, where $d$ is the minimal index of an unequal symbol of $u$ and $v$ such that $d>i$. The error probability is at most $0.1$.  
\end{corollary}

Let us present the algorithm that finds $left_j$. Let us call it $\textsc{LeftBorderSearch(j)}$. The algorithm is based on the Binary search algorithm \cite{cormen2001}.
 Let $Le$ be the left border of the search segment, and $Ri$ be the right border of the search segment. Let $St_i=t[suf_i:n]$ be $i$-th suffix for $i\in\{1,\dots,n\}$.
\begin{itemize}
\item[] \textbf{Step 1.} We assign $Le\gets 1$ and $Ri\gets n$. Let $Llcp\gets\textsc{QLCP}(St_{Le},s^j)$ be the LCP of the first suffix and the searching string. Let $Rlcp\gets\textsc{QLCP}(St_{Ri},s^j)$ be the LCP of the last suffix and the searching string.
\item[] \textbf{Step 2.} If $Llcp<|s^j|$ and $s^j<St_1$, i.e. $Llcp<|s^j|$ and $s^j[Llcp+1]<St_1[Llcp+1]$, then we can say that $s^j$ is less than any suffix of $t$ and is not a prefix of any suffix of $t$. So, we can say that $I_j$ is empty. In that case, we stop the algorithm, otherwise, we continue with Step 3.
\item[] \textbf{Step 3.} If $Rlcp<|s^j|$ and $s^j>St_n$, i.e. $Rlcp<|s^j|$ and $s^j[Rlcp+1]>St_n[Rlcp+1]$, then we can say that $s^j$ greats any suffix of $t$ and is not a prefix of any suffix of $t$. So, we can say that $I_j$ is empty also. In that case, we stop the algorithm, otherwise, we continue with Step 4.
\item[] \textbf{Step 4.} We repeat the next steps while $Ri-Le>1$, otherwise we go to Step 9.
\item[] \textbf{Step 5.} Let $M\gets\lfloor (Le+Ri)/2\rfloor$.
\item[] \textbf{Step 6.} If $Llcp\geq Rlcp$, then we go to Step 7, and to Step 8 otherwise.
\item[] \textbf{Step 7.} Here we compare $LCP(St_L,St_M)$ and $Llcp$. Note that $LCP(St_{Le},St_M)$ can be computed as $\textsc{LCPSuf}(t,Le,M)$ in $O(1)$ query and time complexity. We have one of three options:
\begin{itemize}
\item If $LCP(St_{Le},St_M)>Llcp$, then all suffixes from $St_{Le}$ to $St_M$ are such that $St_M[Llcp+1]=\dots =St_{Le}[Llcp+1]\neq s^j[Llcp+1]$, and they cannot have $s^j$ as a prefix. So, we can assign $Le\gets M$ and not change $Llcp$.  
\item If $LCP(St_{Le},St_M)=Llcp$, then all suffixes from $St_{Le}$ to $St_M$ has at least $Llcp$ common symbols with $s^j$. Let us compute $Mlcp=LCP(St_M, s^j)$ using $\textsc{QLCP}(St_M, s^j, Llcp+1)$. If $Mlcp=|s^j|$, then we can move the right border to $M$ and update $Rlcp$ because we search the leftmost occurrence of $s^j$. So, $Ri\gets M$ and $Rlcp\gets Mlcp$. A similar update of $R$ and $Rlcp$ we do if $St_M[Mlcp+1]>s^j[Mlcp+1]$. If $St_M[Mlcp+1]<s^j[Mlcp+1]$, then $Le\gets M$ and $Llcp\gets Mlcp$.
\item  If $LCP(St_{Le},St_M)<Llcp$, then all suffixes from $St_M$ to $St_R$ cannot have $s^j$ as a prefix. Therefore, we update $Ri\gets M$, and $Rlcp\gets  LCP(St_{Le},St_M)$.
\end{itemize}
After that, we go back to Step 4.
\item[] \textbf{Step 8.}  The step is similar to Step 7,  but we compare $LCP(St_M,St_{Ri})$ with $Rlcp$. We have one of three options:
\begin{itemize}
\item If $LCP(St_M,St_{Ri})>Rlcp$, then $Ri\gets M$ and we do not change $Rlcp$.  
\item If $LCP(St_M,St_{Ri})=Rlcp$, then we compute $Mlcp=LCP(St_M, s^j)$ using $\textsc{QLCP}(St_M, s^j, Rlcp+1)$. Then, we update variables exactly by the same rule as in the second case of Step 7.
\item If $LCP(St_M,St_{Ri})<Rlcp$, then we update $L\gets M$, and $Llcp\gets LCP(St_M,St_{Ri})$.
\end{itemize}
After that, we go back to Step 4.
\item[] \textbf{Step 9.} The result of the search is $Ri$. So, we assign $left_j\gets Ri$.
\end{itemize}
Similarly, we can compute $right_j$ in a procedure $\textsc{RightBorderSearch}(j)$.

The implementation for $\textsc{LeftBorderSearch}(j)$ is presented in Algorithm~\ref{alg:leftSearch}, and the complexity is presented in Lemma \ref{lm:leftSearch}. %(See implementation for $\textsc{LeftBorderSearch(j)} $ that searches $right_j$ is presented in Appendix \ref{apx:rightSearch})

\begin{algorithm}[H]
    \caption{Quantum algorithm for $\textsc{LeftBorderSearch}(j)$ that searches $left_j$.}\label{alg:leftSearch}
    \begin{algorithmic}
        \State $Llcp \gets \textsc{QLCP}(St_{1}, s^j)$, $Rlcp \gets \textsc{QLCP}(St_{n}, s^j)$
        \If{$Llcp = |s^j|$}
 			\State $answer \gets 0$
 		\Elif{$s^j[Llcp+1]<St_1[Llcp+1]$\Then}
        \State $answer \gets -1$
        \Elif{$s^j[Rlcp+1]>St_n[Rlcp+1]$\Then}
        \State $answer \gets -1$
        \Else
     	\State $Le\gets 1$, $Ri\gets n$
        \While{$Ri - Le > 1$}
        \State $M \gets \lfloor (Le + Ri) / 2\rfloor$
        \If{$Llcp\geq Rlcp$}
        		\If{$\textsc{LCPSuf}(Le, M)> Llcp$}
        			\State $Le\gets M$
        		\EndIf
        		\If{$\textsc{LCPSuf}(Le, M)= Llcp$}
        			\State $Mlcp=\textsc{QLCP}(St_M, s^j, Llcp+1)$
        			\If{$Mlcp=|s^j| \Or St_M[Mlcp+1]>s^j[Mlcp+1]$}
        				\State $Ri\gets M$, $Rlcp\gets Mlcp$
        			\Else
        				\State $Le\gets M$, $Llcp\gets Mlcp$
        			\EndIf
			\EndIf        		
        		\If{$\textsc{LCPSuf}(Le, M)< Llcp$}
        			\State $Ri \gets M$, $Rlcp \gets \textsc{LCPSuf}(Le, M)$
			\EndIf	
        		
        \Else
			\If{$\textsc{LCPSuf}(M, Ri)> Rlcp$}
        			\State $Ri\gets M$
        		\EndIf
        		\If{$\textsc{LCPSuf}(M, Ri)= Rlcp$}
        			\State $Mlcp=\textsc{QLCP}(St_M, s^j, Rlcp+1)$
        			\If{$Mlcp=|s^j| \Or St_M[Mlcp+1]>s^j[Mlcp+1]$}
        				\State $Ri\gets M$, $Rlcp\gets Mlcp$
        			\Else
        				\State $Le\gets M$, $Llcp\gets Mlcp$
        			\EndIf
			\EndIf        		
        		\If{$\textsc{LCPSuf}(M, Ri)< Rlcp$}
        			\State $Le \gets M$, $Llcp \gets \textsc{LCPSuf}(M, Ri)$
			\EndIf        		        
        \EndIf
        \EndWhile
       	 \State $answer \gets Ri$
        \EndIf
        \State \Return $answer$
    \end{algorithmic}
\end{algorithm}
\begin{lemma}\label{lm:leftSearch}
Algorithm \ref{alg:leftSearch} implements $\textsc{LeftBorderSearch}(j)$, searches $left_j$ and works with $O(\sqrt{|s^j|\log n}+\log n)$ query complexity, $O(\sqrt{|s^j|\log n}\log|s^j|+\log n)$ time complexity, and at most $0.1$ error probability. 
\end{lemma}
\begin{proof}
Let us consider $Llcp$. It is always increased and never decreased. Assume that during the algorithm, it has values $Llcp_1,\dots, Llcp_d$. We can say that $Llcp_1<\dots< Llcp_d$, and $d\leq \log_2 n$ because Binary search do at most $\log_2 n$ steps. Additionally, we can say that $Llcp_d\leq |s^j|$. The function $\textsc{QLCP}$ increases $Llcp$. Query complexity of changing $Llcp$ from $Llcp_i$ to $Llcp_{i+1}$ is $O(\sqrt{Llcp_{i+1}-Llcp_{i}})$ due to Corollary \ref{cr:lcp}. The time complexity is \[O(\sqrt{Llcp_{i+1}-Llcp_{i}}\log(Llcp_{i+1}-Llcp_{i}))=O(\sqrt{Llcp_{i+1}-Llcp_{i}}\log(|s^j|)).\]
So, the total query complexity is \[\sqrt{Llcp_{1}} + \sum\limits_{i=1}^{d-1}O(\sqrt{Llcp_{i+1}-Llcp_{i}})=O\left(\sqrt{d(Llcp_{1} + \sum\limits_{i=1}^{d-1}(Llcp_{i+1}-Llcp_{i}))}\right)\] due to Cauchy--Bunyakovsky--Schwarz inequality. At the same time, $Llcp_{1} + \sum\limits_{i=1}^{d-1}(Llcp_{i+1}-Llcp_{i})=Llcp_{d}\leq |s^i|$. Finally, the complexity of all changing of $Llcp$ is $O(\sqrt{d|s^j|})=O(\sqrt{|s^j|\log n})$. Using the same technique, we can show that the time complexity is $O(\sqrt{|s^j|\log n}\log|s^j|)$
Similarly, we can see that the complexity of all changing of $Rlcp$ is $O(\sqrt{d|s^j|})=O(\sqrt{|s^j|\log n})$. All other parts of the step of Binary search work with $O(1)$ query and time complexity, including $\textsc{LCPSuf}$ due to Lemma~\ref{lm:lcp2-arr}. They give us additional $O(\log n)$ complexity.

Note that each invocation of $\textsc{QLCP}$ has an error probability $0.1$. Therefore, $O(\log n)$ invocations of the procedure have an error probability that is close to $1$. At the same time, the whole Algorithm \ref{alg:leftSearch} is a sequence of First-One Search procedures such that next invocation of the First-One Search uses results of the previous invocations. Due to \cite{k2014}, such sequence can be converted to an algorithm with the same total complexity and error probability $0.1$.    
\Endproof
\end{proof}

The whole algorithm is presented in Algorithm \ref{alg:search}, and complexity is discussed in Theorem \ref{th:allSearch}

\begin{algorithm}[H]
    \caption{Quantum algorithm for the Multiple String Matching Problem.}\label{alg:search}
    \begin{algorithmic}
    
		\State $\textsc{ConstructSuffixArray}(t)$
		\State $\textsc{PreprocessingForLCP}(t)$
        \For{$j\in\{1,\dots,m\}$}
          \State $left_j\gets \textsc{LeftBorderSearch}(j)$
           \State $right_j\gets \textsc{RightBorderSearch}(j)$
           \If{$left_j=-1\Or right_j=-1$}
           \State $I_j\gets ()$
           \Else
           \State $I_j\gets (suf_{left_j},\dots,suf_{right_j})$
           \EndIf
        \EndFor
        \State \Return $(I_1,\dots,I_m)$
    \end{algorithmic}
\end{algorithm}
% \section{Quantum Version of LZ78} \label{sec:lz78}
\begin{theorem}\label{th:allSearch}
Algorithm \ref{alg:search} solves the Multiple String Matching Problem and works with $O(n+\sqrt{mL\log n}+m\log n)$ query complexity, $O(n+\sqrt{mL\log n}\log b+m\log n)$ time complexity and at most $0.1$ error probability. 
\end{theorem}
\begin{proof}
Firstly, complexity of  $\textsc{ConstructSuffixArray}(t)$ and $\textsc{PreprocessingForLCP}(t)$ is $O(n)$ due to Lemmas \ref{lm:suf-arr}, \ref{lm:lcp-arr}, and \ref{lm:lcp2-arr}.

Query complexity of $\textsc{LeftBorderSearch}(j)$ and $\textsc{RightBorderSearch}(j)$ is $O(\sqrt{|s^j|\log n}+\log n)$, and time complexity is $O(\sqrt{|s^j|\log n}\log |s^j|+\log n)$.
So, the total query complexity of these procedures for all $j\in\{1,\dots,m\}$ is
\[O\left(m\log n +\sum_{j=1}^m\sqrt{|s^j|\log n}\right)=O\left(m\log n +\sqrt{\log n}\sum_{j=1}^m\sqrt{|s^j|}\right)=\]\[
O\left(m\log n +\sqrt{\log n\cdot m\sum_{j=1}^m|s^j|}\right)=O\left(m\log n +\sqrt{mL\log n}\right)\]
We get such a result because of Cauchy--Bunyakovsky--Schwarz inequality. Similarly, time complexity is $O\left(m\log n +\sqrt{mL\log n}\log b\right)$ since $b=max\{|s^1|,\dots,|s^m|\}$.

The total query complexity is $O(n+\sqrt{mL\log n}+m\log n)$, and time complexity is $O(n+\sqrt{mL\log n}\log d+m\log b)$.

Note that, each invocation of $\textsc{LeftBorderSearch}(j)$ and $\textsc{RightBorderSearch}(j)$ has error probability $0.1$. Therefore, $O(m)$ invocations of the procedure have an error probability that is close to $1$. At the same time, the algorithm is a sequence of First-One Search procedures. Due to \cite{k2014}, such sequence can be converted to an algorithm with the same total complexity and error probability $0.1$.    
\Endproof
\end{proof}

Let us present the lower bound for the problem.
\begin{theorem}\label{th:allSearchlb}
The lower bound for query complexity of the Multiple String Matching Problem is $\Omega(n+L)$ in the classical case, and $\Omega(n+\sqrt{mL})$ in the quantum case. 
\end{theorem}
\begin{proof}
Let us consider a binary alphabet for simplicity. Let $s^1=``1''$. Then, the complexity of searching positions of $s^1$ in $t$ is at least as hard as an unstructured search of all positions of $1$ among $n$ bits.  Due to \cite{bbbv1997}, classical query complexity of this problem is $\Omega(n)$ and quantum query complexity is $\Omega(\sqrt{nt})$, where $t$ is the number of occurrences $1$ among $n$ bits. In the worst case, it is also $\Omega(n)$.

Let us consider $t$ as a string of $0$s. In that case, if we have any symbol $1$ in a dictionary string, then this string does not occur in $t$. We should find all such strings for forming results for these strings. Number of such strings is at most $m$. The size of the search space is $L$. Then, the complexity of searching the answer is at least as hard as an unstructured search of at least $m$ positions of $1$ among $L$ bits. Due to \cite{bbbv1997}, classical query complexity of this problem is $\Omega(n)$ and quantum query complexity is $\Omega(\sqrt{mL})$

So, the total lower bound is $\Omega(max(n,L))=\Omega(n+L)$ in the classical case and $\Omega(max(\sqrt{n},\sqrt{mL}))=\Omega(\sqrt{n}+\sqrt{mL})$ in the quantum case.
\Endproof
\end{proof}
\section{Conclusion}\label{sec:concl}

In the paper, we present a quantum algorithm for the Multiple String Matching Problem that works with $O(n+\sqrt{mL}+m\log n)$ query complexity and error probability $0.1$. It is better than the classical counterparts if $O(m)$ strings of the dictionary have at least $\omega(\log n)$ length. In that case, $\sqrt{mL\log n}=o(L)$ and $m\log n=o(L)$. Similar situation for time complexity. We obtain speed up if $m\log n(\log b)^2=o(L)$, or at lest $O(m)$ strings have length $\omega(\log n\log\log n)$.

In the quantum case the lower bound is $\Omega(n+\sqrt{mL})$ and upper bound is  $O^*(n+\sqrt{mL})$. They are equal up to a log factor. The open question is to develop a quantum algorithm with complexity that is equal to the lower bound.

%\section*{Acknowledgments}

%We sincerely thank Aliya Khadieva for her help.
%, Nail Nurmeev, and Ilnar Zinnatullin for useful discussions. 

%todo add funds in the final version
%todo add KFTI in the final version

\bibliographystyle{splncs04}
\bibliography{tcs}
\end{document}